\documentclass[letterpaper, 10pt, conference]{ieeeconf}
\IEEEoverridecommandlockouts
\usepackage{graphics}
\usepackage{epsfig}
\usepackage{times}
\usepackage{amsmath,amssymb,amsfonts}
\usepackage{mathrsfs}
\usepackage{algorithm}
\usepackage{algpseudocode}
\usepackage[utf8]{inputenc}
\usepackage{xcolor}

\newtheorem{theorem}{Theorem}
\newtheorem{remark}{Remark}

\title{\LARGE \bf
Data-Driven Predictive Control for Stochastic Descriptor Systems: An Innovation-Based Approach Handling Non-Causal Dependencies}

\author{Yunxiang Ma, Yibo Wang, Zhongmei Li, and Chao Shang,~\IEEEmembership{Member,~IEEE} 
\thanks{This work was supported by the National Natural Science Foundation of China under Grants 62373211 and 62327807.}%
\thanks{Yunxiang Ma, Yibo Wang, and Chao Shang are with the Department of Automation, Tsinghua University, Beijing 100084, China. Email: {\tt\small \{mayx23,wyb21\}@mails.tsinghua.edu.cn, c-shang@tsinghua.edu.cn}}%
\thanks{Zhongmei Li is with State Key Laboratory of Industrial Control Technology, Ministry of Education, East China University of Science and Technology, Shanghai 200237, China. Email: {\tt\small zhongmeili@ecust.edu.cn}}
}

\begin{document}

\maketitle
\thispagestyle{empty}
\pagestyle{empty}

\begin{abstract}
Descriptor systems arise naturally in real-world applications governed by algebraic constraints, such as power networks, robotics and chemical processes. When a descriptor model contains a nontrivial nilpotent block, the discrete-time input--output map may be improper: the current output depends on future inputs and, in the stochastic case, on future noise terms. This letter proposes a data-driven predictive control framework for stochastic descriptor systems that handles these non-causal dependencies without explicitly identifying system matrices. The key idea is to split fast subsystem into noise-driven and input-driven parts, and then combine the former with the slow subsystem such that an innovation-driven Kalman filter can be appropriately defined to reformulate the stochastic descriptor system into an innovation-driven form. Based on this, a new behavioral system representation is derived, which inspires a data-driven innovation-based multi-step output predictor and a practical Inno-DeePC algorithm that enables data-driven predictive control design without known system matrices while implicitly handling algebraic constraints. Numerical experiments on a DC microgrid demonstrate the effectiveness of the proposed approach.
\end{abstract}

\section{Introduction}

Descriptor systems provide a natural framework for modeling  real-life physical systems in which dynamic laws coexist with algebraic constraints, such as conservation laws in power networks, contact forces in robotics, and mass balances in chemical processes~\cite{luenberger1977dynamic,dai1989singular}. Unlike standard state-space models admitting proper transfer functions, descriptor systems, also known as differential algebraic equations (DAEs), can lead to improper input--output dependencies. After discretization, this structure appears as dependence between current output and future input or noise realizations. Such a non-causal dependency of the descriptor representation is \textit{not} a violation of physical causality, but is meaningful in local subsystems of large-scale networks, where algebraic coupling and port choices can expose improper channels~\cite{li2024descriptor}.

Predictive control design of dynamical systems typically hinges on an accurate predictive model, obtained either from first-principles modeling or offline identification. This is rather non-trivial for descriptor systems because in addition to differential dynamics, one must specify algebraic constraints concurrently. For example, in power networks, capacitor and inductor dynamics coexist with Kirchhoff constraints, and certain port choices, capacitor-voltage loops, or inductor-current cutsets can introduce a higher-index nilpotent structure in a local descriptor model~\cite{tischendorf1999topological,estevez2000structural}, making exact first-principles modeling notoriously difficult. Besides, accurately identifying descriptor systems from data remains a significant challenge~\cite{he2023identification}.

The difficulties in modeling real-world complex systems have stipulated the recent prevalence of direct data-enabled predictive control (DeePC) approaches \cite{coulson2019data}. They are primarily built upon the behavioral system theory~\cite{willems2005note}. For standard state-space systems, existing variants of DeePC focus on robustness, stability, stochasticity, noisy data, data informativity, and the connection between direct and indirect formulations~\cite{berberich2021data,breschi2023data,waarde2024behavioral,pan2023stochastic,van2023data,dorfler2023bridging}. For deterministic descriptor systems, tailored DeePC methods have already been put forward in~\cite{schmitz2022willems,schmitz2022case}; however, they do not address stochastic disturbances or measurement noise.


For proper state-space systems, the innovation-based DeePC (Inno-DeePC) was proposed in \cite{wang2024innovation} to handle uncertainty by augmenting inputs with an estimated innovation sequence, giving a novel data-driven output predictor that resembles a Kalman filter (KF). However, it is non-trivial to extend this idea to descriptor systems since the innovation form cannot be defined any longer due to the non-causal dependency. To address this challenge, this letter proposes a new data-driven predictive control framework tailored for stochastic descriptor systems. The main contributions are as follows:
\begin{itemize}
	\item We decompose the quasi-Weierstra\ss{} fast subsystem into noise-driven and input-driven parts, and then combine the former with the slow subsystem such that a steady-state KF can be appropriately defined to reformulate the stochastic descriptor system into an innovation-driven form and make multi-step ahead predictions.
	\item Based on the innovation-driven reformulation, we obtain a new behavioral representation of stochastic descriptor systems, and develop a new data-driven output predictor with innovations as extra inputs. To reasonably estimate innovations from input-output data, a new least-squares strategy is suggested.
    \item A practical predictive control algorithm for stochastic descriptor systems is put forward. Numerical experiments on a simulated DC microgrid validates its efficacy.
\end{itemize}

\emph{Notation:}
For a positive semidefinite matrix $W$, $\|z\|_W^2 \triangleq z^\top W z$. Signal sequences are written as $z_{[a,b]} = \operatorname{col}\{ z(a),\ldots,z(b) \}$, where $\operatorname{col}\{ X_1,\ldots,X_q \} \triangleq [X_1^\top,\ldots,X_q^\top]^\top$. The symbols $I_n$ and $0$ denote the identity matrix of size $n$ and a zero matrix of compatible dimension, respectively. The symbols $\succeq$ and $\succ$ denote positive semidefiniteness and positive definiteness. The Moore--Penrose pseudo-inverse, Kronecker product, and Gaussian distribution are denoted by $(\cdot)^\dagger$, $\otimes$, and $\mathcal{N}(\mu,\Sigma)$, respectively.

\section{Preliminaries and Problem Formulation}

\subsection{System Description and Assumptions}

Consider a discrete-time linear time-invariant (LTI) stochastic descriptor system:
\begin{equation}
	\label{descriptor_system}
	 \left \{ \begin{aligned}
		\mathcal{E} x(k+1) &= A x(k) + B u(k) + w(k), \\
		y(k) &= C x(k) + D u(k) + v(k),
	\end{aligned} \right .
\end{equation}
where $x(k) \in \mathbb{R}^n$, $u(k) \in \mathbb{R}^m$, $y(k) \in \mathbb{R}^p$, and the matrix $\mathcal{E} \in \mathbb{R}^{n \times n}$ may be singular, i.e., $\operatorname{rank}(\mathcal{E}) < n$, which distinguishes descriptor systems from standard state-space models. The noise sequences $w(k) \sim \mathcal{N}(0, Q)$ and $v(k) \sim \mathcal{N}(0, R)$ are assumed to be mutually independent Gaussian white noise with $Q \succeq 0$ and $R \succ 0$. We assume that the matrix pencil $z\mathcal{E} - A$ is regular, i.e., $\det(z\mathcal{E}-A)$ is not the zero polynomial, which ensures the descriptor dynamics are well posed for every consistent initial condition~\cite{dai1989singular}.

By the quasi-Weierstra\ss{} decomposition~\cite{dai1989singular,schmitz2022willems}, there exist nonsingular matrices $S,P\in\mathbb{R}^{n\times n}$ such that
\begin{equation*}
	S\mathcal{E}P=\begin{bmatrix} I_q & 0 \\ 0 & N \end{bmatrix},\qquad
	SAP=\begin{bmatrix} A_1 & 0 \\ 0 & I_r \end{bmatrix},
\end{equation*}
where $N\in\mathbb{R}^{r\times r}$ is nilpotent with index $s$ (i.e., $N^s=0$ and $N^{s-1}\neq 0$) and $q+r=n$. Applying the state coordinate change $z(k)=P^{-1}x(k)=\operatorname{col}\{ z_{\rm slow}(k),z_{\rm fast}(k) \}$ with $z_{\rm slow}\in\mathbb{R}^q$ and $z_{\rm fast}\in\mathbb{R}^r$ and the noise transformation $\bar w(k)=Sw(k)=\operatorname{col}\{ w_1(k),w_2(k) \}$, the stochastic descriptor system \eqref{descriptor_system} decouples into a slow subsystem with states $z_{\rm slow}$ of dimension $q$ and a fast algebraic subsystem with states $z_{\rm fast}$ of dimension $r$. Then we further split the fast states $z_{\rm fast}(k)$ into two components, i.e. $z_{\rm fast}(k)=z_2(k)+z_3(k)$, where $z_2(k)$ and $z_3(k)$ are driven by noise and inputs, respectively. Letting $z_1(k) = z_{\rm slow}(k)$, then the stochastic descriptor system \eqref{descriptor_system} can be equivalently transformed into:
\begin{equation}
	\label{eq_three_part_descriptor}
    \left \{
	\begin{aligned}
		&z_1(k+1) = A_1 z_1(k) + B_1 u(k) + w_1(k),\\
		&N z_2(k+1) = z_2(k) + w_2(k),\\
		&N z_3(k+1) = z_3(k) + B_3 u(k),\\
		&y(k) = C_1 z_1(k)+C_2 z_2(k)+C_3 z_3(k)+D u(k)+v(k),
	\end{aligned} \right .
\end{equation}
where $C_2 = C_3$. By recursive substitution and using $N^s=0$, we have
\begin{equation}
	\label{eq_fast_solution}
	z_2(k)=-\sum_{i=0}^{s-1}N^i w_2(k+i),
	z_3(k)=-\sum_{i=0}^{s-1}N^i B_3 u(k+i),
\end{equation}
so there exist non-causal dependencies between $z_2(k)$ and future noise realizations, and between $z_3(k)$ and future inputs.

\begin{remark}
	For descriptor systems with $s<2$ or $r=0$, the input--output map has no non-causal dependency: any algebraic contribution, if present, involves only current input or noise terms, and thus can be tackled using generic DeePC methods for standard LTI systems. For descriptor systems with $s\geq 2$, however, $z_2(k)$ and $z_3(k)$ are intrinsically relevant to future noise and future input terms, respectively. This is the descriptor non-causal dependency to be tackled in this letter.
\end{remark}

\subsection{Data-Driven Trajectory Characterization}

We next revisit the data-driven representation of the descriptor system~\eqref{descriptor_system} in the noise-free case, i.e., $w(k)=0$ and $v(k)=0$. This hinges on the so-called R-controllability~\cite{dai1989singular}, where ``R'' refers to reachability of the finite dynamics. In the quasi-Weierstra\ss{} coordinates, the noise-free counterpart of system~\eqref{descriptor_system} is said to be R-controllable if
\begin{equation}
	\label{eq_R_ctrl_pencil}
	\operatorname{rank}\bigl([B_1,\ A_1 B_1,\ \ldots,\ A_1^{q-1} B_1]\bigr) = q.
\end{equation}
Based on this, the following result states that all admissible input--output trajectories of a deterministic descriptor system can be expressed using a single sufficiently informative trajectory, without explicitly knowing $\{E,A,B,C,D\}$.
\begin{theorem}
	\label{thm_FL_descriptor}

    \textit{(Behavioral representation of noise-free descriptor system~\cite{schmitz2022willems})}
	Consider the \emph{deterministic} (noise-free) counterpart of system~\eqref{descriptor_system} with $w(k)=0$ and $v(k)=0$, and assume it is regular and R-controllable with a $q$th-order slow subsystem and nilpotency index $s$. Let $\{u^d(k), y^d(k)\}_{k=0}^{T-1}$ be an input--output data trajectory of this noise-free system with $u^d(\cdot)$ persistently exciting of order $L+q+s-1$.\footnote{A $T$-length sequence $u_{[0,T-1]}$ is said to be persistently exciting of order $L$ if its $L$-depth block Hankel matrix $\mathcal{H}_L(u_{[0,T-1]})$ has full row rank.}
	Then $(\bar{u},\bar{y})$ is a valid input--output trajectory of the same noise-free system if and only if there exists $\alpha \in \mathbb{R}^{T-s-L+2}$ such that
    \begin{equation*}
        \begin{bmatrix}
            \mathcal{H}_{L}(u^d_{[0,T-s]}) \\
			\mathcal{H}_{L}(y^d_{[0,T-s]})
        \end{bmatrix}\alpha
        =
        \begin{bmatrix} \bar{u} \\ \bar{y} \end{bmatrix}.
    \end{equation*}

\end{theorem}

\vspace{0.1in}

\begin{remark}
	The Hankel matrices are built from the truncated sequences $u^d_{[0,T-s]}$ and $y^d_{[0,T-s]}$ because the last $s-1$ samples lack the future-input information required by the fast subsystem. The usual minimum data-length requirement for $L$ is implied by the persistency-of-excitation condition above.
\end{remark}

\section{Innovation-Based Equivalent Reformulation and Inno-DeePC}

In this section, we formalize data-driven predictor and control design for stochastic descriptor systems~\eqref{descriptor_system}. Unlike generic state-space models, descriptor systems with $s>1$ are directly amenable to a KF-based predictor for handling uncertainty. To circumvent this, we make use of \eqref{eq_three_part_descriptor} and decompose the output as $y(k)=y_{12}(k)+y_3(k)$, where
\begin{equation}
	\label{eq_y12_decompose}
	\begin{aligned}
		y_{12}(k)&\triangleq C_1z_1(k)+C_2z_2(k)+v(k),\\
		y_3(k)&\triangleq C_3z_3(k)+D u(k).
	\end{aligned}
\end{equation}
The former describes output variations caused by the slow subsystem, the non-causal dependence on noise, and the measurement noise, whereas the latter captures the non-causal dependence on inputs and the feedthrough term. Note that the direct feedthrough $Du(k)$ is grouped into $y_3(k)$ rather than $y_{12}(k)$. We show that for $y_{12}(k)$ a KF-based predictor can be defined, which enables us to design a novel innovation-based behavioral predictor for the stochastic descriptor system~\eqref{descriptor_system}.

\subsection{Noise-Buffered Kalman Predictor}

Recall that $w_1(k)$ and $w_2(k)$ in~\eqref{eq_three_part_descriptor} are jointly correlated since both originate from $w(k)$ via $\bar w(k)=Sw(k)$. Let $\bar{Q}\triangleq SQS^\top$ denote the covariance of $\operatorname{col}\{ w_1(k),w_2(k) \}$, and let $r_w=\operatorname{rank}(\bar Q)$. Factor $\bar{Q}=\Gamma\Gamma^\top$ with $\Gamma=[\Gamma_1^\top,\Gamma_2^\top]^\top\in\mathbb{R}^{n\times r_w}$, where $\Gamma_1\in\mathbb{R}^{q\times r_w}$ and $\Gamma_2\in\mathbb{R}^{r\times r_w}$, so that $w_1(k)=\Gamma_1\varepsilon(k)$ and $w_2(k)=\Gamma_2\varepsilon(k)$ with $\varepsilon(k)\sim\mathcal{N}(0,I_{r_w})$. Since $z_2(k)$ depends on $\{\Gamma_2 \varepsilon(k),\ldots,\Gamma_2 \varepsilon(k+s-1)\}$, we introduce a noise buffer $\eta(k) \triangleq [\varepsilon(k)^\top,\ldots,\varepsilon(k+s)^\top]^\top \in \mathbb{R}^{r_w(s+1)}$ with shift-register dynamics
\begin{equation*}
	\eta(k+1) = A_\eta \eta(k) + B_\eta \varepsilon(k+s+1),
\end{equation*}
where $A_\eta$ is the block upward-shift matrix and $B_\eta=[0;\ldots;0;I_{r_w}]$. Define $M\triangleq[I_r,N,\ldots,N^{s-1},0]\in\mathbb{R}^{r\times r(s+1)}$ and $S_2\triangleq I_{s+1}\otimes\Gamma_2\in\mathbb{R}^{r(s+1)\times r_w(s+1)}$. It then follows that $z_2(k)=-MS_2\eta(k)$. Consequently, the dynamics of $y_{12}(k)$ is governed by
\begin{equation}
	\label{eq_augmented_system}
	\left \{ \begin{aligned}
		\xi(k+1) &= A_\xi \xi(k)+B_\xi u(k) +G_\xi \varepsilon(k+s+1),\\
		y_{12}(k) &= C_\xi \xi(k)+v(k),
	\end{aligned} \right .
\end{equation}
where $\xi(k)\triangleq[z_1(k)^\top,\eta(k)^\top]^\top\in\mathbb{R}^{n_\xi}$ with $n_\xi=q+r_w(s+1)$, and
\begin{equation*}
	\begin{gathered}
		A_\xi=\begin{bmatrix} A_1 & \Gamma_1 S_0 \\ 0 & A_\eta \end{bmatrix},
		\quad S_0=[I_{r_w},0,\ldots,0],\quad
		B_\xi=\begin{bmatrix} B_1 \\ 0 \end{bmatrix},\\
		G_\xi=\begin{bmatrix} 0 \\ B_\eta \end{bmatrix},\quad
		C_\xi=\begin{bmatrix} C_1 & -C_2 M S_2 \end{bmatrix}.
	\end{gathered}
\end{equation*}
Notice that the future noise term $\varepsilon(k+s+1)$ is independent of $\xi(k)$, $u(k)$, and $v(k)$. Hence, if~\eqref{eq_augmented_system} is minimal, there exists a steady-state KF of the form
\begin{equation}
	\label{eq_inno_form}
	\left \{ \begin{aligned}
		\hat{\xi}(k+1|k) &= A_\xi \hat{\xi}(k|k-1) + B_\xi u(k) + K_p e(k), \\
		\hat{y}_{12}(k) &= C_\xi \hat{\xi}(k|k-1),
	\end{aligned} \right .
\end{equation}
where $K_p$ is the steady-state Kalman gain, $\hat{y}_{12}(k)$ is the one-step prediction of $y_{12}(k)$, and 
\begin{equation}
	\label{eq_inno}
    e(k) = y_{12}(k) - \hat{y}_{12}(k)
\end{equation}
is the innovation, i.e., the one-step prediction error. Due to the Gaussianity of $\varepsilon(k+s+1)$ and $v(k)$, it then follows from~\cite{kailath1968innovations} that $e(k)$ in steady states is a Gaussian white noise sequence with covariance $\Sigma_e=C_\xi P_\infty C_\xi^\top+R$, where $P_\infty$ is the steady-state prediction-error covariance of $\xi(k)$. Given $\hat{\xi}(k+1|k)$ and future inputs $u(k+\tau),~\tau \ge 1$, one can further attain multi-step forward predictions of $y_{12}(k)$:
\begin{equation}
	\label{eq_multi_step}
	\left\{
	\begin{aligned}
		\hat{\xi}(k+\tau+1|k) &= A_\xi \hat{\xi}(k+\tau|k) + B_\xi u(k+\tau), \\
		\hat{y}_{12}(k+\tau|k) &= C_\xi \hat{\xi}(k+\tau|k),
	\end{aligned}
	\right.
\end{equation}
which can be understood as propagating the predictor multiple times with future innovations replaced by their conditional zero mean. 

Leveraging~\eqref{eq_inno_form} and~\eqref{eq_inno}, we can eventually recast the stochastic descriptor system~\eqref{descriptor_system}, equivalently~\eqref{eq_three_part_descriptor}, into an equivalent quasi-Weierstra\ss{} reformulation that is explicitly driven by both the system input $u(k)$ and the innovation $e(k)$:
\begin{equation}
	\label{eq_recast_descriptor}
    \left \{
	\begin{aligned}
		&\hat{\xi}(k+1|k) = A_\xi \hat{\xi}(k|k-1) + B_\xi u(k) + K_p e(k), \\
		&N z_3(k+1) = z_3(k) + B_3 u(k),\\
		&y(k) = C_\xi \hat{\xi}(k|k-1)+C_3 z_3(k)+D u(k)+e(k),
	\end{aligned} \right .
\end{equation}
where $e(k)$ essentially captures the uncertainty due to the noise terms in~\eqref{eq_three_part_descriptor}. The output equation in~\eqref{eq_recast_descriptor} can also be written as
\begin{equation}
    \label{eq_full_output_y12_y3}
    y(k)=\hat{y}_{12}(k)+y_3(k)+e(k).
\end{equation}

\subsection{Innovation-Based Behavioral Representation}

Based upon the innovation-driven equivalent reformulation~\eqref{eq_recast_descriptor}, we now arrive at a new behavioral representation of the stochastic descriptor system~\eqref{descriptor_system}, which extends the deterministic Theorem~\ref{thm_FL_descriptor} to the stochastic setting.

\begin{theorem}[Innovation-Based Fundamental Lemma]
	\label{thm_FL_inno}
	Consider the stochastic descriptor system~\eqref{descriptor_system}, from which an input--output data trajectory $\{u^d(k),y^d(k)\}_{k=0}^{T-1}$ is collected, and let $\{e^d(k)\}_{k=0}^{T-1}$ denote the associated innovation sequence generated by the steady-state KF~\eqref{eq_inno_form}. Suppose that the recast innovation-driven system~\eqref{eq_recast_descriptor} is R-controllable, i.e., the slow subsystem~\eqref{eq_inno_form} is controllable from the augmented input $(u,e)$, and that the combined offline driving signal $\tilde{u}^d(k)\triangleq\operatorname{col}(u^d(k),e^d(k))$ is persistently exciting of order $L+n_\xi+s-1$. Then any valid $L$-length trajectory $(\bar{u},\bar{y},\bar{e})$ of the recast system~\eqref{eq_recast_descriptor} can be expressed as
	\begin{equation}
		\begin{bmatrix}
				U^d \\
				Y^d \\
				E^d
		\end{bmatrix} g
		=
		\begin{bmatrix}
			\bar{u} \\ \bar{y} \\ \bar{e}
		\end{bmatrix}
        \label{eq: 12}
	\end{equation}
	for some $g \in \mathbb{R}^{T-s-L+2}$, where $U^d=\mathcal{H}_{L}(u^d_{[0,T-s]})$, $Y^d=\mathcal{H}_{L}(y^d_{[0,T-s]})$ and $E^d=\mathcal{H}_{L}(e^d_{[0,T-s]})$ are Hankel matrices.
\end{theorem}

\emph{Proof:}
The recast system~\eqref{eq_recast_descriptor} can be viewed as a deterministic descriptor system in quasi-Weierstra\ss{} form, with an $n_\xi$-order slow subsystem and driven by the augmented external input $\tilde{u}(k)=\operatorname{col}\{ u(k),e(k) \}$. Under the stated assumptions, Theorem~\ref{thm_FL_descriptor} applies directly to~\eqref{eq_recast_descriptor} with this augmented input, which yields the claimed Hankel representation. $\hfill\blacksquare$

\begin{remark}
	The persistency-of-excitation order $L+n_\xi+s-1$ in Theorem~\ref{thm_FL_inno} is imposed on the joint driving signal $\tilde u^d(k)$ that stacks $u^d(k)$ and $e^d(k)$. This implies that $u^d(k)$ itself has to be persistently exciting of order $L+n_\xi+s-1$. Because $n_\xi>q$, this requirement on $u^d(k)$ is more demanding than the deterministic case in Theorem~\ref{thm_FL_descriptor}. The corresponding requirement on $e^d(k)$, on the other hand, is mild owing to its white-noise property.
\end{remark}

\begin{remark}[Data-driven innovation estimation] 
\label{rmk_inno_est}
In practice, it is impossible to know the ground truth of the innovation sequence $\{e^d(k)\}_{k=0}^{T-1}$. In~\cite{wang2024innovation}, a non-parametric VARX fitting approach was suggested to estimate innovations of stochastic LTI systems from an input--output data trajectory $\{u^d(k),y^d(k)\}$. However, this approach does not apply straightforwardly to~\eqref{eq_inno_form} since the output $y_{12}(k)$, which is associated with $e(k)$, is in fact unmeasurable. As a remedy, we propose to rely on~\eqref{eq_full_output_y12_y3} to derive an approximate VARX relation:
\begin{equation}
	\label{eq_arx}
	y(k) \approx \sum_{i=1}^{\ell}\Theta_y^{(i)} y(k-i)
	+\sum_{i=-s+1}^{\ell}\Theta_u^{(i)} u(k-i)+e(k),
\end{equation}
	where the $s-1$ future-input regressors $u(k+1),\ldots,u(k+s-1)$ (corresponding to $i=-s+1,\ldots,-1$), which are absent in~\cite{wang2024innovation}, are added to describe the non-causal dependency in $y_3(k)$. The current input $u(k)$ and past inputs $u(k-1),\ldots,u(k-\ell)$ are also included as standard regressors. Here $\ell$ is the autoregressive order, which shall be sufficiently large such that the truncation error $\|(A_\xi-K_p C_\xi)^{\ell}\|$ becomes negligible~\cite{chiuso2007role}. Given an input--output data trajectory $\{u^d(k),y^d(k)\}$, we perform least squares to estimate the coefficients $\{\hat{\Theta}_y^{(i)},\hat{\Theta}_u^{(i)}\}$ in~\eqref{eq_arx}, and then regard the fitting residuals as the estimated innovations $\hat{e}^d(k)$. We emphasize that this least-square estimation procedure does not amount to identifying $\{E,A,B,C,D\}$ from data, since the estimated coefficients $\{\hat{\Theta}_y^{(i)},\hat{\Theta}_u^{(i)}\}$ do not correspond to any particular descriptor realization of $\{E,A,B,C,D\}$.
\end{remark}

\subsection{Inno-DeePC for Stochastic Descriptor Systems}

Having estimated $\{\hat{e}^d(k)\}_{k=0}^{T-s}$, one can construct a Hankel matrix $\hat{E}^d=\mathcal{H}_{L}(\hat{e}^d_{[0,T-s]})$ and insert it into \eqref{eq: 12} as a substitute of $E_d$. In order to predict evolution of system outputs in a future horizon, it is necessary to further split \eqref{eq: 12} into past and future sections. Specifically, three Hankel matrices $U^d$, $Y^d$, $\hat{E}^d$ shall be split as:
\begin{equation}
    U^d = \begin{bmatrix}
        U_p \\ U_f
    \end{bmatrix}, ~Y^d = \begin{bmatrix}
        Y_p \\ Y_f
    \end{bmatrix}, ~\hat{E}^d = \begin{bmatrix}
        \hat{E}_p \\ \hat{E}_f
    \end{bmatrix},
\end{equation}
where $U_p,Y_p,\hat E_p$ are past Hankel matrices with $L_p$ block rows, and $U_f,Y_f,\hat E_f$ are future Hankel matrices with $L_f = L - L_p$ block rows. In principle, the past horizon $L_p$ shall satisfy $L_p \ge n_\xi + s - 1$ in order to uniquely determine the initial condition \cite[Lemma 1]{schmitz2022willems}, whereas the future horizon $L_f$ shall satisfy $L_f\geq s$ so as to supply the $s-1$ terminal input samples required by the descriptor fast subsystem. In this way, we arrive at an approximate version of behavioral relation \eqref{eq: 12}:
\begin{equation}
	\label{eq_FL_true_inno}
	\operatorname{col}\{U_p,U_f,Y_p,Y_f,\hat E_p,\hat E_f\}\,g
	=\operatorname{col}\{u_p,u_f,y_p, y_f,\hat e_p,\hat e_f\},
\end{equation}
where $u_p = u_{[t-L_p,t-1]}$, $u_f = u_{[t,t+L_f-1]}$, and $y_p, y_f,\hat e_p,\hat e_f$ can be defined analogously.
Due to the usage of data-driven estimates $\{ \hat E_p,\hat E_f \}$, the realized innovation trajectories $\hat e_p,\hat e_f$ no longer coincide with $e_p, e_f$ generated from an ideal KF \eqref{eq_inno_form} and are thus distinguished using $\hat{\cdot}$. Following the conditional-mean prediction principle \cite{wang2024innovation} and the spirit of \eqref{eq_multi_step}, we can predict the future output $\hat y_f$ by setting future innovations in \eqref{eq_FL_true_inno} to zero, i.e., $\hat e_f=0$. This gives rise to an implementable innovation-based data-driven multi-step output predictor:
\begin{equation}
	\label{eq_FL_est_inno}
	\operatorname{col}\{U_p,U_f,Y_p,Y_f,\hat E_p,\hat E_f\}\,g
	=\operatorname{col}\{u_p,u_f,y_p,\hat y_f,\hat e_p,0\}.
\end{equation}
To eliminate the equality constraint $\hat E_f g=0$, we apply null-space projection: let $\hat E_f^\bot$ be a matrix whose columns form an orthonormal basis for the null space of $\hat E_f$ (assumed full row rank). Writing $g=\hat E_f^\bot h$ and substituting it into four block rows of~\eqref{eq_FL_est_inno} gives
\begin{equation}
	\label{eq_Pi_h}
	\Pi h = \operatorname{col} \{ u_p,\,u_f,\,y_p,\,\hat e_p \},
\end{equation}
where $\Pi\triangleq\operatorname{col} \{ U_p,\,U_f,\,Y_p,\,\hat E_p \} \hat E_f^\bot$, and $u_f$ is the planned input sequence containing the $s-1$ terminal samples required by the descriptor fast subsystem. To suppress the effect of inaccurate innovation estimation, we take the minimum-norm solution of~\eqref{eq_Pi_h}:
\begin{equation*}
	h_{\mathrm{pinv}}=\Pi^\dagger \operatorname{col}\{ u_p,u_f,y_p,\hat e_p \}.
\end{equation*}
This yields a new multi-step predictor of future output
\begin{equation}
	\label{eq_yf_predictor}
	\hat{y}_f = Y_f \hat{E}_f^\bot\, h_{\mathrm{pinv}},
\end{equation}
which is affine in future control actions $u_f$ since $h_{\mathrm{pinv}}$ depends linearly on $(u_p,y_p,\hat e_p,u_f)$.

Based upon the proposed multi-step predictor \eqref{eq_yf_predictor}, the predictive control design problem can be formulated, e.g. as a convex quadratic program (QP) in $u_f$, with reference $r$, polytopic constraint sets $\mathcal{U}$ and $\mathcal{Y}$, and weights $Q_y\succeq 0$ and $R_u\succ 0$:
\begin{equation}
	\label{eq_deepc}
	\begin{aligned}
		\min_{u_f}\quad & \sum_{k=0}^{L_f-1}\|\hat{y}(t+k)-r(t+k)\|_{Q_y}^2 + \|u(t+k)\|_{R_u}^2 \\
		\mathrm{s.t.}\quad & \hat{y}_{[t,t+L_f-1]} = \hat{y}_f\ \text{given by}\ \eqref{eq_yf_predictor}, \\
		& u(t+k)\in\mathcal{U},~ \hat{y}(t+k)\in\mathcal{Y},~ k=0,\cdots,L_f-1.
	\end{aligned}
\end{equation}
The pseudo-inverse formulation~\eqref{eq_yf_predictor} delivers a minimum-norm predictor consistent with the innovation-based prediction principle in~\cite{wang2024innovation}, thereby reducing the need for the ad-hoc regularization commonly used in standard DeePC under noisy data. After solving~\eqref{eq_deepc}, the receding-horizon input $u^\star(t)$ is applied to the system; the online innovation is updated as the one-step prediction residual $\hat{e}(t)=y(t)-\hat{y}(t)$, and the past innovation window is shifted accordingly.

\begin{remark}
	The data-driven predictive control problem \eqref{eq_deepc} effectively bypasses the first-principle modeling or system identification procedure. It also implicitly encodes the algebraic constraint structure in \eqref{descriptor_system} by means of the behavioral representation \eqref{eq: 12} and the multi-step predictor \eqref{eq_yf_predictor}.
\end{remark}

\begin{remark}[Knowledge of nilpotency index $s$]
	Our Inno-DeePC algorithm requires the nilpotency index $s$ or its upper bound to be known. In power systems, this information is often available from the circuit or interconnection topology, e.g., from capacitor-voltage loops or inductor-current cutsets in electrical networks~\cite{dai1989singular,li2024descriptor}. Also, in robotics and multibody dynamics, when mechanisms are modeled in Cartesian coordinates with holonomic joint or contact constraints, the resulting equations are known to be DAEs with $s = 3$, since the position-level constraints must be differentiated to obtain acceleration-level dynamics and constraint forces~\cite{negrut2005hht}.
    
\end{remark}

Algorithm~\ref{alg_inno_deepc} summarizes the offline/online implementation of our Inno-DeePC method for stochastic descriptor systems. Note that innovation estimation is only performed offline and not needed in online predictive control.

\begin{algorithm}[h]
	\caption{Inno-DeePC for Stochastic Descriptor Systems}
	\label{alg_inno_deepc}
	\begin{algorithmic}[1]
		\Statex \textbf{Input:} Data $\{u^d(k),y^d(k)\}_{k=0}^{T-1}$; upper bound of $s$; horizons $L_p$, $L_f$ with $L_f \geq s$; VARX order $\ell$; weights $Q_y, R_u$.
		\State \textbf{Offline:} Fit~\eqref{eq_arx} by least squares and set the residuals as $\hat e^d(k)$.
		\State \textbf{Offline:} Construct $U_p,U_f,Y_p,Y_f,\hat E_p,\hat E_f$; compute $\hat E_f^\bot$ and $\Pi^\dagger$; verify that $\Pi$ has full column rank.
		\State \textbf{Online:} Form $u_p(t), y_p(t), \hat e_p(t)$ and solve~\eqref{eq_deepc} using the predictor~\eqref{eq_yf_predictor}.
		\State \textbf{Online:} Apply $u^\star(t)$; update $\hat e(t)=y(t)-\hat y(t)$; shift the past-innovation window; and repeat.
	\end{algorithmic}
\end{algorithm}

\section{Numerical Simulation}

We validate the proposed Inno-DeePC algorithm for stochastic descriptor systems on the control task of a 48\,V DC microgrid~\cite{dragicevic2016dc} featuring algebraic constraints and an improper output channel. As discussed in Section~I, the local plant seen by an individual converter can exhibit an improper transfer function due to algebraic coupling, making this a representative testbed for the proposed method.

Consider the four-bus DC microgrid shown in Fig.~\ref{fig_sys}, comprising two voltage source converters (VSC\,1 and VSC\,2), interconnecting cables with series resistance and inductance, two bus capacitors, and a resistive load. The descriptor state is $x=[V_1,V_2,V_3,V_4,i_{12},i_{23},i_{24}]^\top\in\mathbb{R}^7$, the input is $u=[u_1,u_2]^\top\in\mathbb{R}^2$ involving current signals, and the measured output is $y=[V_1,V_3,V_4]^\top\in\mathbb{R}^3$. The main circuit parameters are $C_{\mathrm{bus},1}=2.2$\,mF, $C_{\mathrm{bus},4}=1.5$\,mF, $L_{12}=0.5$\,mH, $L_{23}=0.8$\,mH, $L_{24}=0.6$\,mH, $R_{12}=0.10\,\Omega$, $R_{23}=0.15\,\Omega$, $R_{24}=0.12\,\Omega$, and $R_L=60\,\Omega$.

Kirchhoff's laws yield the continuous-time descriptor model $E\dot{x}=Ax+Bu$, $y=Cx$, where $E=\operatorname{diag}(C_{\mathrm{bus},1},0,0,C_{\mathrm{bus},4},L_{12},L_{23},L_{24})$ is singular with $\operatorname{rank}(E)=5$. Rows~2 and~3 correspond, respectively, to the algebraic Kirchhoff current law (KCL) constraint at Bus\,2 and to the current-source relation $i_{23}=u_2$. The quasi-Weierstra\ss{} decomposition produces a slow subsystem of dimension $q=5$ and a fast subsystem of dimension $r=2$ with nilpotency index $s=2$.

\begin{figure}[h]
	\centering
	\includegraphics[width=0.8\columnwidth, trim=15 15 15 15, clip]{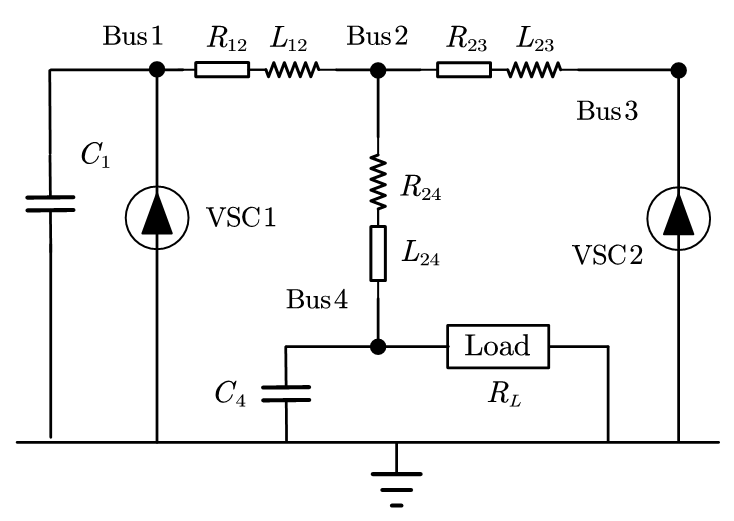}
	\caption{Topology of the 48\,V DC microgrid. VSC\,1 injects a current $u_1$ at Bus 1, while the current source $u_2$ at Bus 3 models a dispatchable distributed generator.}
	\label{fig_sys}
\end{figure}

The second measured output $V_3$ constitutes an improper channel: by Kirchhoff's voltage law (KVL), $V_3=V_2-R_{23}i_{23}-L_{23}\dot{i}_{23}$ depends on $\dot{u}_2$ since $i_{23}=u_2$. After discretization with a first-order hold (FOH) at sampling period $h=0.1$\,s, this manifests as the fast state depending on $u(k+1)$, which requires the VARX regressor to include $s-1=1$ future-input term as prescribed by~\eqref{eq_arx}.

An offline dataset of $T=300$ samples is collected using a persistently exciting input composed of pseudo-random binary sequence signals, sinusoids, and Gaussian noise around the nominal operating point $u^{\mathrm{nom}}=[5.0,2.5]^\top$\,A. Process noise $w(k)\sim\mathcal{N}(0,0.03^2 I_7)$ and measurement noise $v(k)\sim\mathcal{N}(0,0.6^2 I_3)$ yield an output signal-to-noise ratio (SNR) of approximately $33$\,dB. The VARX order is set to $\ell=15$, producing $N_e=284$ effective innovation samples with regressor dimension $n_\phi=79$. Controller horizons are $L_p=12$ and $L_f=21 \geq s=2$), with cost weights $Q_y=I_3$ and $R_u=0.05\,I_2$.

We compare our method against Regularized DeePC (Reg-DeePC)~\cite{coulson2019data} and Subspace Predictive Control (SPC)~\cite{favoreel1999spc}. \textbf{Reg-DeePC} follows~\cite{coulson2019data} with regularization parameter $\lambda_g=50$ and equality constraints enforcing past input--output consistency. \textbf{SPC} constructs a linear predictor via oblique projection of the future output onto the past and future input--output subspaces and solves the resulting unconstrained QP in $u_f$~\cite{favoreel1999spc}. All three controllers share identical cost weights.

The closed-loop experiment spans $150$ control steps ($15$\,s), preceded by a warm-up phase of $L_p=12$ random-input steps. A setpoint change is introduced at $k=82$ ($8.2$\,s): during $k\in[12,81]$ the reference is computed from $u^{\mathrm{nom}}=[5.0,2.5]^\top$\,A, while during $k\in[82,161]$ it switches to the reduced-load operating point $u^{\mathrm{step}}=[4.0,1.8]^\top$\,A. All three controllers are subject to identical noise realizations.

\begin{figure}[h]
	\centering
	\includegraphics[width=0.85\columnwidth]{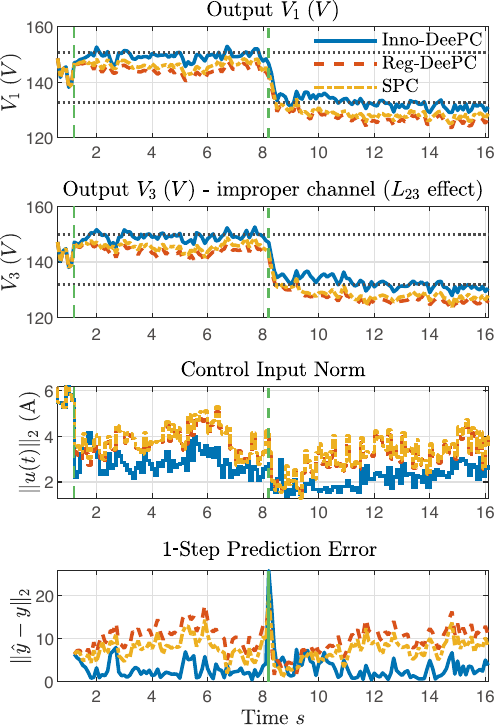}
	\caption{Closed-loop performance comparison on the 48\,V DC microgrid. Dashed black lines denote the setpoints; the first vertical green line marks the end of the warm-up phase, and the second marks the setpoint switch.}
	\label{fig_sim}
\end{figure}

Fig.~\ref{fig_sim} displays the output trajectories, control effort, and prediction errors for all three data-driven control design methods. During the first setpoint phase, Inno-DeePC (solid blue) tracks the reference voltages with the smallest steady-state fluctuation, whereas Reg-DeePC (dashed red) and SPC (dash-dotted yellow) exhibit noticeable tracking errors. After the setpoint switches at $t=8.2$\,s, Inno-DeePC settles to the new reference with the shortest transient among the three methods. Notably, the improper output $V_3$, which depends on $\dot{u}_2$ through the cable inductance $L_{23}$, is well regulated by Inno-DeePC, confirming that the algorithm correctly handles the non-causal output dependency.

The bottom panel of Fig.~\ref{fig_sim} shows that Inno-DeePC consistently attains the lowest prediction error throughout the experiment. Using $R^2=1-\sum_k\|y(k)-\hat y(k)\|^2/\sum_k\|y(k)-\bar y\|^2$, where $\bar y$ is the sample mean, the one-step-ahead prediction scores are $0.917$ (Inno-DeePC), $0.727$ (SPC), and $0.531$ (Reg-DeePC). The substantially higher $R^2$ demonstrates that the innovation-based multi-step predictor~\eqref{eq_yf_predictor} produces improved output predictions by exploiting the statistical structure of the descriptor system noise.

\section{Conclusion}

This letter developed an innovation-based data-driven predictive control framework for stochastic descriptor systems with non-causal dependencies induced by nilpotent descriptor blocks. Technically we decomposed the quasi-Weierstra\ss{} fast subsystem into noise-driven and input-driven parts, and then combined the former with the slow subsystem such that a steady-state KF can be appropriately defined to reformulate the stochastic descriptor system into an innovation-driven form. Based on this, we derived a new behavioral representation of stochastic descriptor systems, and develop a new data-driven output predictor as well as predictive control design scheme with innovations as extra inputs. To reasonably estimate innovations from input-output data, a new least-squares strategy was suggested. Experiments on a DC microgrid demonstrated improved tracking and prediction accuracy relative to Reg-DeePC and SPC. Future work will address closed-loop stability guarantees and extensions to nonlinear descriptor systems.

\bibliographystyle{ieeetr}
\bibliography{references}

\end{document}